\documentclass[conference,letterpaper]{IEEEtran}
\usepackage[utf8]{inputenc} 
\usepackage[T1]{fontenc}
\usepackage{url}
\usepackage{ifthen}
\usepackage{cite}
\usepackage[cmex10]{amsmath} 
\usepackage{amsthm}
\usepackage{amsmath}
\usepackage{amsfonts}
\usepackage{amssymb}
\usepackage{color}
\usepackage{graphicx} 
\usepackage{caption}
\usepackage{stackrel}
\usepackage{booktabs}
\usepackage{enumitem}
\usepackage{balance}
\usepackage[utf8]{inputenc}
\theoremstyle{plain}
\newtheorem{theorem}{Theorem}

\newtheorem{corollary}{Corollary}
\newtheorem{lemma}{Lemma}
\newtheorem{definition}{Definition}

\newcommand{\In}{{\rm In}}
\newcommand{\Out}{{\rm Out}}

\newcommand{\cS}{\mathcal{S}}

\title{Towards an Operational Definition of Group Network Codes}
\author{\IEEEauthorblockN{Fei Wei}
\IEEEauthorblockA{University at Buffalo\\
feiwei@buffalo.edu}
\and
\IEEEauthorblockN{Michael Langberg}
\IEEEauthorblockA{University at Buffalo \\
mikel@buffalo.edu}
\and
\IEEEauthorblockN{Michelle Effros}
\IEEEauthorblockA{California Institute of Technology\\
effros@caltech.edu}
}

\begin{document}

\maketitle


\begin{abstract}

Group network codes are a generalization of linear codes that have seen several studies over the last decade. 
When studying network codes, operations performed at internal network nodes called local encoding functions, are of significant interest.
While local encoding functions of linear codes are well understood (and of operational significance), no similar operational definition exists for group network codes. 
To bridge this gap, we study the connections between group network codes and a family of codes called Coordinate-Wise-Linear (CWL) codes. 
CWL codes generalize linear codes and, in addition, can be defined locally (i.e., operationally). 
In this work, we study the connection between CWL codes and group codes from both a local and global encoding perspective. 
We show that Abelian group codes can be expressed as CWL codes and, as a result, they inherit an operational definition.

\end{abstract}


\section{Introduction}\label{sec:intro}

Network coding is a well studied communication paradigm on noiseless networks that enables network nodes to encode information before subsequent transmissions, e.g., \cite{koetter2003algebraic,li2003linear,effros2003linear,jaggi2004linear,yeung2008information}.
In the network coding literature, it is common to distinguish between {\em local} and {\em global} encoding functions.
A local encoding function $\phi_{le}$ for network edge $e=(u,v)$ determines the information transmitted on $e$ as a function of the incoming information to the tail node $u$ of $e$. 
A global encoding function $\phi_{ge}$ for edge $e$ determines the information transmitted on $e$ as a function of the network source random variables $(X_i:i \in \cS)$.
(Detailed definitions for the concepts above and those that appear below appear in Section~\ref{sec:model}.)
Local encoding functions capture the {\em operational} aspect of network coding, in the sense that they characterize the distributed encoding process performed locally at network nodes.
Global encoding functions capture how source information is processed throughout the network, in the sense that they explicitly tie the information transmitted on network edges with the information present at network sources.
In the context of acyclic networks, given a collection of local encoding functions one can inductively derive the corresponding global encoding functions, e.g., \cite{yeung2008information}.

\begin{figure*}[t!]
    \centering
    \captionsetup{justification=centering}
    \includegraphics[draft=false,width=0.8\textwidth]{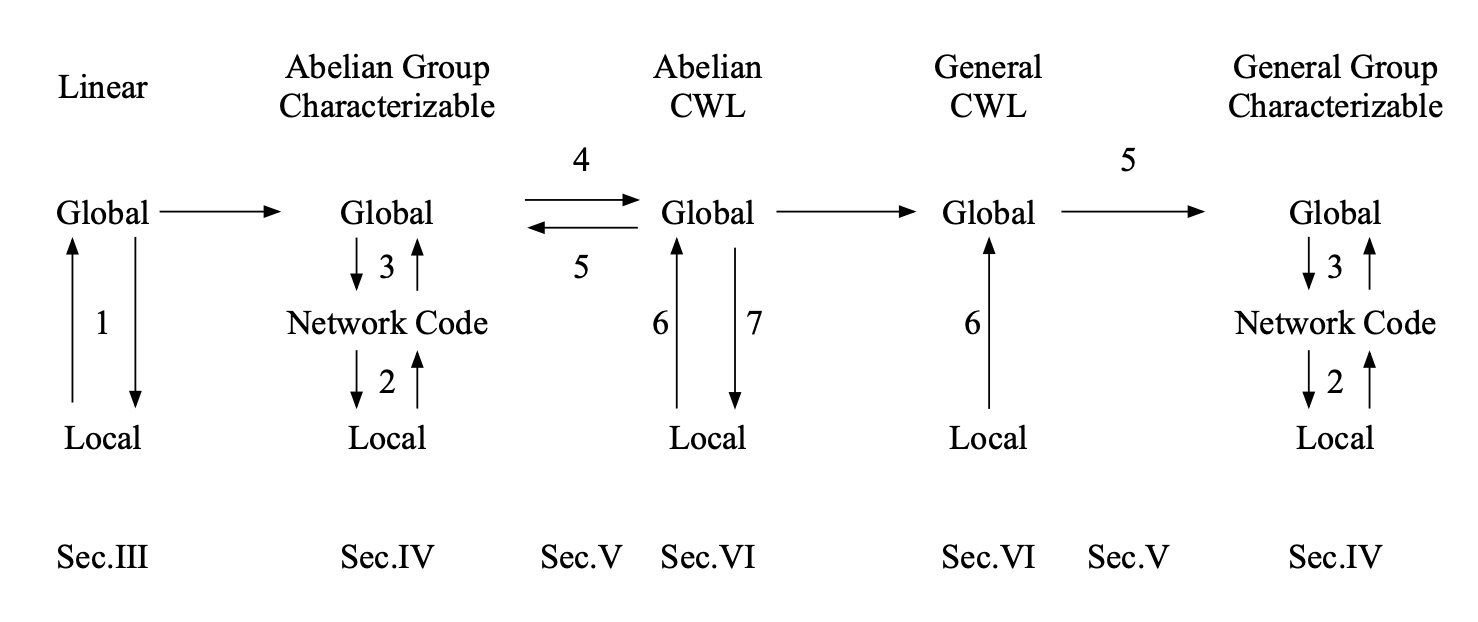}
    \caption{An outline of the relations among different types of functions. The number marked on each implication represents the corresponding theorem in this work proving the implication (with corresponding sections below).} 
    \label{fig:functions}
\end{figure*}

{\em Linear codes} are an efficient and widely used method for the encoding and decoding of information.
In the context of network coding, linear encoding has been extensively studied, e.g., \cite{koetter2003algebraic,li2003linear,effros2003linear,jaggi2004linear,ho2006random,yeung2008information,ho2010equivalence}.
Operationally, local encoding functions that linearly combine the incoming messages to a given edge yield efficient communication schemes for use in practice, e.g., \cite{yeung2007avalanche,dimakis2010network,chaudhuri2019secure}.
Local encoding functions that are linear give way to linear global encoding functions, implying that terminals receive linear combinations of the source random variables, a fact found very useful in the analysis of network coding schemes.
Although linear codes suffice to obtain the multicast capacity \cite{li2003linear}, for general network coding instances, with multiple sources and multiple terminals, linear codes fall short of achieving capacity \cite{dougherty2005insufficiency}.

{\em Group network codes}, first defined in \cite{chan2005optimality}, are a generalization of linear codes.
Roughly speaking, in linear codes edge messages are characterized by linear subspaces of the source vector space, while in group codes both source messages and edge messages are characterized by certain co-sets of subgroups of a given ambient group $G$. 
Group network codes do not suffer from the sub-optimality of linear codes, as any achievable network coding rate vector can be {\em approximated} by a group code \cite{chan2002relation,chan2005optimality}.
While linear codes may be defined locally, group codes lack such an operational definition.
{\em In this work, we seek an operational definition for group network codes - one that will broaden our understanding of group codes and potentially allow the design of low complexity local encoding functions.}

Towards this end, we study a family of codes called Component-Wise-Linear (CWL) codes \cite{wei2017effect}, which, as group codes, generalize linear codes, albeit from an operational perspective. 
A local encoding function $\phi_{le}$ for edge $e$ is CWL if one can associate a group structure with each incoming edge to $e$ and with the edge $e$ itself such that the mapping expressed by $\phi_{le}$ is a homomorphism. 
One can similarly define global CWL functions. 
Linear codes are shown to be CWL by choosing the corresponding groups to again be subspaces of the source vector space.
Further, the network codes defined by (global) CWL functions are group codes \cite{wei2019locala},\cite{wei2019local}.
{\em In this work we address the complementary question asking whether group codes can be represented operationally through CWL functions.}

The results of this work are summarized in Figure~\ref{fig:functions}.
For linear, CWL, and group codes, we study the notion of both local and global functions, some of which have not been explicitly defined before.
We compare between the local and global variants of linear, CWL, and group encoding functions and analyze their relation.

Our work is structured as follows.
In Section~\ref{sec:model}, we present our model and the definitions of linear, CWL, and group codes.
In Section~\ref{sec:linear}, we study the relationship between local and global encoding in the context of linear codes. 
The results presented in Section~\ref{sec:linear} are folklore and given here for completeness. 
In Section~\ref{sec:group}, we study group network codes.
We define a notion of local and global encoding and study the relationship between them.
In Section~\ref{sec:groupandcwl}, we study the relationship between group network codes and CWL codes.
We distinguish between Abelian and non-Abelian group structures.
In Section~\ref{sec:cwl}, we investigate the relationship between locally and globally defined CWL codes.
Finally, we conclude in Section~\ref{sec:conclusion}. 

One of the main consequences of our analysis lies in the combination of Theorems~\ref{theorem:group2}, \ref{theorem:groupCWL4}, and \ref{theorem:CWL7} (see Figure~\ref{fig:functions}), which collectively show that Abelian group codes can be represented operationally by Abelian CWL codes, and thus the former inherit the operational aspects of the latter (see Corollary~\ref{col:only}).

The proofs of several claims appear in the Appendix.


\section{Model and definition} \label{sec:model}

We denote the set $\{1,\dots,k\}$ by $[k]$ for any positive integer $k$. 
Given a random variable $X$, we use the calligraphic letter $\mathcal{X}$ to represent its alphabet and use lower case $x$ to represent a realization of $X$.
Given an index set $\alpha$, $X_\alpha$ is the collection of random variables $(X_a: a\in\alpha)$ with support $\mathcal{X}_\alpha = \prod_{a\in\alpha} \mathcal{X}_a$ equal to the Cartesian product of $\{\mathcal{X}_a: a\in\alpha\}$.
For a singleton set, we may omit brackets, for example writing $a$ in place of $\{a\}$.


\subsection{Network Instance}

A network instance $\mathcal{I} = (\mathcal{N},\mathcal{S},\mathcal{T},\mathcal{M})$ includes a directed acyclic error-free network $\mathcal{N} = G(\mathcal{V,E})$ with nodes (also referred to as vertices) $\mathcal{V}$ and edges $\mathcal{E} \subset \mathcal{V}\times\mathcal{V}$, a set of sources $\mathcal{S}\subset \mathcal{V}$, a set of terminals $\mathcal{T}\subset \mathcal{V}$ and a demand matrix $\mathcal{M}$, where $m_{st} = 1$ if and only if terminal $t \in \mathcal{T}$ demands source $s \in \mathcal{S}$.
Each edge $e=(u,v)\in\mathcal{E}$ represents an error-free point-to-point link from node $u$ to node $v$ with edge capacity $R_e > 0$.
For each node $v\in\mathcal{V}$, we denote the set of incoming and outgoing edges of node $v$ as $ \In(v) = \{(v_1,v):(v_1,v)\in\mathcal{E}\}$ and $\Out(v) = \{(v,v_1):(v,v_1)\in\mathcal{E}\}$, respectively. 
Without loss of generality, we assume that there are no incoming edges for any source $s\in\mathcal{S}$ and no outgoing edges for any terminal $t\in\mathcal{T}$, giving, $\mathcal{S}\cap\mathcal{T}=\phi$.


\subsection{Network Code}

Let $\mathcal{I}$ be a network instance.
A \textit{network code} of block length $n$ on $\mathcal{I}$ is defined by a set of random variables $\{X_f : f\in \mathcal{S}\cup\mathcal{E}\}$ as follows. 
Each source $i\in\mathcal{S}$ with rate $R_i$ independently generates source message $X_i$ uniformly at random over the alphabet $\mathcal{X}_i = [2^{nR_i}]$.
Each edge $e \in \mathcal{E}$ carries edge message $X_e$ with alphabet $\mathcal{X}_e = [2^{n R_e}]$. 

For any edge $e=(u,v)\in \mathcal{E}$, random variable $X_e$ is determined by the incoming random variables $X_{\In(u)}$.
Namely, with each edge we can associate a \textit{local encoding function} $\phi_{le}: \mathcal{X}_{\In(u)} \mapsto \mathcal{X}_e$ that takes as its input the message tuple $X_{\In(u)}$ of random variables associated with incoming edges $\In(u)$.
The edge message $X_e$ equals the evaluation of $\phi_{le}$ on its input, giving $X_e=\phi_{le}(X_{\In(u)})$.
Edges that leave source node $s\in\mathcal{S}$ have corresponding local encoding functions that take the source information $X_s$ as input.
At any terminal node $t\in\mathcal{T}$, the decoding function $\phi_{t}:\mathcal{X}_{\In(t)} \mapsto \hat{\mathcal{X}}_t$ takes as input incoming messages $X_{\In(t)}$ and emits the reproduction $\hat{X}_t$ of the demanded source message $X_t$.
Decoding is considered successful if $\hat{X}_t=X_t$.

As for each edge $e=(u,v)$ $X_e$ is a function of $X_{\In(u)}$, we inductively obtain that $X_e$ is a function of the source message tuple $X_\mathcal{S}$ as well.
We can thus associate with each edge $e$ a \textit{global encoding function} $\phi_{ge}: \mathcal{X}_{\mathcal{S}} \mapsto \mathcal{X}_e,$ giving $X_e=\phi_{le}(X_{\In(u)})=\phi_{ge}(X_\mathcal{S})$.


\subsection{Linear Functions}

\begin{definition}[Linear Function]\label{def:linearFunc}
A surjective function $\phi : \mathcal{X}_a \mapsto \mathcal{X}_b$ is called linear if and only if there exists a vector space $V$ over a finite base field $\mathbb{F}$ with subspaces $V_a$ and $V_b$, and a matrix $T$ where $\mathcal{X}_a = V_a$ and $\mathcal{X}_b = V_b$, such that $X_b= \phi(X_a)=X_a T$.
\end{definition}


\subsection{Group Characterizable Random Variables}


\begin{definition}[Group Characterizable Random Variables] \label{def:groupRV}
Random variables $\{X_a:a\in\mathcal{A}\}$ are group characterizable if and only if there exists a finite group $G$ with subgroups $\{G_a:a\in\mathcal{A}\}$ such that given an element $g_r$ chosen uniformly at random from $G$ (referred to as the ``uniform element $g_r$''), it holds for all $a \in \mathcal{A}$ that $X_a = g_r G_a$.
The alphabet $\mathcal{X}_a$ of $X_a$, equals the set of left cosets of $G_a$ in $G$.
The group $G$ and subgroups $\{G_a:a\in\mathcal{A}\}$ are called a group characterization of $\{X_a:a\in\mathcal{A}\}$.
\end{definition}

Group characterizable random variables were introduced in \cite{chan2002relation}.
For any $\alpha\subseteq \mathcal{A}$, if $\{X_a:a\in\mathcal{A}\}$ are group characterizable, then by our definitions, the following properties hold for $X_\alpha$:
\begin{enumerate}
    \item $X_\alpha$ is quasi-uniform over $\mathcal{X}_\alpha$ and $support(X_\alpha)=\{(g G_a: a\in\alpha): g\in G\}$.
    \item Let $G_\alpha=\cap_{a\in\alpha} G_a$. 
    By our definitions, $X_\alpha = (g_r G_1,\dots,g_r G_{|\alpha|})$ where $g_r$ is the uniform element in $G$.
    As $\cap_{a\in\alpha} (g_r G_a) = g_r (\cap_{a\in\alpha} G_a) = g_r G_\alpha$, $X_\alpha= g_r G_\alpha$.
    Namely, $X_\alpha$ is equivalent to the random variable distributed uniformly over $G/G_\alpha$.
    \item By the properties above, $H(X_\alpha)=\log \frac{|G|}{|G_\alpha|}$.
\end{enumerate}

\begin{definition}[Group Characterizable Function] \label{def:groupFunc}
Consider a set of random variables $\{X_a: a\in\mathcal{A}\}$. 
For any $\alpha\subset \mathcal{A}$ and $b\in \mathcal{A}$, a surjective function $\phi: X_\alpha \mapsto X_b$ is called group characterizable if and only if there exists a finite group $G$ and subgroups $\{G_f:f\in \alpha \cup b\}$ which are a group characterization of random variables $\{X_f: f\in\alpha\cup b\}$.
\end{definition}

Notice that if $\phi: X_\alpha \mapsto X_b$ is a group characterizable function with representation $\{G_f: f\in \alpha\cup b\}$ then $G_\alpha\subset G_b$.
This follows from the fact that $\log\frac{|G|}{|G_{\alpha \cup b}|}= H(X_{\alpha}, X_b) =H(X_\alpha)=\log\frac{|G|}{|G_{\alpha}|}$.

By the definition above, for any edge $e=(u,v)\in\mathcal{E}$, a local encoding function $\phi_{le}$ is group characterizable if and only if random variables $\{X_f:f\in {\In(u)}\cup e\}$ are group characterizable.
Similarly, a global encoding function $\phi_{ge}$ is group characterizable if and only if random variables $\{X_f: f\in\mathcal{S}\cup e\}$ are group characterizable.

\begin{definition}[Group Network Code] \label{def:groupNC}
A network code $\{X_f: f\in\mathcal{S}\cup\mathcal{E}\}$ is called a group network code if and only if there exists a finite group $G$ and subgroups $\{G_f:f\in\mathcal{S}\cup\mathcal{E}\}$ that characterize $\{X_f: f\in\mathcal{S}\cup\mathcal{E}\}$.
A group network code is called Abelian if $G$ is Abelian.
\end{definition}

In some parts of our discussion, we consider encoding functions that are as group characterizable in a ``consistent manner'', as defined next. 

\begin{definition}[Consistent Group Characterization]  \label{def:groupConsis}
Let $\{\phi_f:f\in\mathcal{F}\}$ be a collection of group characterizable functions over $\{X_a: a\in \mathcal{A}\}$.
For each function $\phi_f$, we denote the index set of the input and output random variables of $\phi_f$ by $\mathcal{A}_f\subset \mathcal{A}$. 
(For example, if $\phi_{f}: \mathcal{X}_{\alpha_f} \mapsto \mathcal{X}_{b_f}$, $\mathcal{A}_{f}=\alpha_f \cup b_f$.)
We say that functions $\{\phi_f:f\in\mathcal{F}\}$ have a consistent group characterization if and only if there exists a finite group $G$ and subgroups $\{G_i: i\in \cup_{f\in\mathcal{F}} \mathcal{A}_f\}$ such that for each $f\in\mathcal{F}$, $\phi_f$ is a group characterizable by $G$ and $\{G_i:i\in \mathcal{A}_f\}$.
\end{definition}

By Definition~\ref{def:groupConsis}, (i) each encoding function in the collection is group characterizable within a common group $G$, and (ii) the group characterizations of any two functions in the collection that involve the same random variable $X_f$ use the same subgroup $G_f$. 
Group characterizations provide us a way to observe the dependency among random variables, and the consistency described here (and below) serves as a tool to allow the comparisons studied in this work between different families of coding functions.


\subsection{Coordinate-Wise-Linear (CWL) Functions}

\begin{definition}[Coordinate-Wise-Linear Function]
A surjective function $\phi: \mathcal{X}_\alpha \mapsto \mathcal{X}_b$ is called coordinate-wise-linear (CWL) if and only if there exist finite groups $\{H_f: f\in \alpha \cup b\}$, with group operation $\stackrel{f}{\circ}$ defined on $H_f$, where $H_f= \mathcal{X}_f$, such that $\phi$ is a group homomorphism from $\prod_{a\in\alpha} H_a$ to $H_{b}$. 
Namely, for any $(x_1,\ldots,x_{|\alpha|}),(x_1',\ldots,x_{|\alpha|}')\in \mathcal{X}_\alpha$ it holds that
\begin{equation*}
\begin{aligned}
    &\phi(x_1\stackrel{1}{\circ} x_1', \ldots, x_{|\alpha|} \stackrel{|\alpha|}{\circ} x_{|\alpha|}')\\ =&\phi(x_1, \dots, x_{|\alpha|}) \stackrel{b}{\circ}  \phi(x_1', \dots, x_{|\alpha|}')\\
\end{aligned}
\end{equation*}
In addition, a CWL function is called Abelian if the groups involved in the definition are Abelian.
\end{definition}

\begin{definition}[Consistent CWL Functions]\label{def:cwlConsis}
Let $\{\phi_f:f\in \mathcal{F}\}$ be a collection of CWL functions over $\{X_a:a\in \mathcal{A}\}$ where for each function $\phi_f: \mathcal{X}_{\alpha_f}\mapsto \mathcal{X}_{b_f}$, the index set of the input and output random variables is defined as $\mathcal{A}_f=\alpha_f\cup b_f$.
We say that functions $\{\phi_f:f\in \mathcal{F}\}$ are consistent CWL functions if and only if there exist finite groups $\{H_i:i\in \cup_{f\in\mathcal{F}} \mathcal{A}_f\}$ where $\mathcal{X}_i = H_i$ for $i\in \cup_{f\in\mathcal{F}} \mathcal{A}_f$, such that for each $f\in\mathcal{F}$, $\phi_f$ is a group homomorphism from $\prod_{a\in\alpha_f} H_a$ to $H_{b_f}$. 
\end{definition}

In order to distinguish the groups involved in CWL functions and group characterizations, the former are denoted by ``$H$'' and the latter by ``$G$''. 
We use $\mathbf{i}_a$ to denote the identity element in the group $G_a$ and $g_a^{-1}$ to denote the inverse of element $g_a\in G_a$.
Let $(G,\cdot),(H,\circ)$ be finite groups, we use ``$G\times H$'' to denote the external direct product of $G$ and $H$. 
For $(g,h),(g',h')\in G\times H$, we define $(g,h)(g',h')=(g\cdot g', h\circ h')$.


\section{Linear Network Code}\label{sec:linear}

Since vector spaces and subspaces are Abelian groups, it holds immediately by our definitions that a linear function is also (Abelian) group characterizable and CWL, which explains the first unmarked arrow in Figure~\ref{fig:functions}.
We start by studying local and global variants of linear codes.
The results of this section are folklore and given here for completeness.
The proof of Theorem~\ref{theorem:linear1} appears in the Appendix.


\begin{theorem}\label{theorem:linear1}
Let $\mathcal{I}$ be a network coding instance and $\{X_f: f\in\mathcal{S}\cup\mathcal{E}\}$ be a network code on $\mathcal{I}$.
Then $\{X_f: f\in\mathcal{S}\cup \mathcal{E}\}$ can be obtained through global encoding functions that are linear if and only if $\{X_f: f\in \mathcal{S}\cup\mathcal{E}\}$ can be obtained by local encoding functions that are linear.
\end{theorem}


\section{Group Network Code} \label{sec:group}

\begin{lemma} \label{lemma:groupunionofRV}
Let $\{X_f: f\in \mathcal{A}\}$ and $\{X_f: f\in \mathcal{B}\}$ be sets of random variables characterized by a finite group $G$ and subgroups $\{G_f: f\in \mathcal{A}\}$ and $\{G_f: f\in \mathcal{B}\}$, respectively. 
Assume that for any single random variable $X$ that appears as $X_a$ for $a\in \mathcal{A}$ and $X_b$ for $b\in \mathcal{B}$, we have $G_a= G_b$.
Then the set of random variables $\{X_f: f\in \mathcal{A}\cup \mathcal{B}\}$ is group characterizable.
\end{lemma}


Lemma~\ref{lemma:groupunionofRV} follows directly from our definitions.
With Lemma~\ref{lemma:groupunionofRV}, we can prove the following two theorems, the detailed proofs appear in the Appendix.


\begin{theorem} \label{theorem:group1}
Let $\mathcal{I}$ be a network coding instance and $\{X_f: f\in\mathcal{S}\cup\mathcal{E}\}$ be a network code on $\mathcal{I}$. 
The network code is a group network code if and only if the local encoding functions on every edge $e\in\mathcal{E}$ have a consistent group characterization.
\end{theorem}


\begin{theorem} \label{theorem:group2}
Let $\mathcal{I}$ be a network coding instance and $\{X_f: f\in\mathcal{S}\cup\mathcal{E}\}$ be a network code on $\mathcal{I}$.
The network code is a group network code if and only if the global encoding functions on every edge $e\in\mathcal{E}$ have a consistent group characterization.
\end{theorem}

Theorems \ref{theorem:group1} and \ref{theorem:group2} show that a code with a consistent group structure, whether given in local or global form, implies a group structure on all random variables in the given code.
While a local encoding function $\phi_{le}$ on an edge $e=(u,v)$ can be group characterizable, the local encoding operation in $\phi_{le}$ is based on the relation between the subgroups $\{G_{f}:f \in In(u)\cup e\}$ and the ambient group $G$. 
In what follows we seek to better understand this relation in an attempt to give it a concrete operational interpretation.


\section{Between Group Characterizable functions and CWL functions} \label{sec:groupandcwl}

\begin{theorem}\label{theorem:groupCWL4}
Let $\mathcal{I}$ be a network coding instance and $\{X_f: f\in\mathcal{S}\cup\mathcal{E}\}$ be a network code on $\mathcal{I}$. 
For any $e\in\mathcal{E}$, if the global encoding function $\phi_{ge}$ is Abelian group characterizable, then $\phi_{ge}$ is Abelian CWL.
\end{theorem}

\begin{proof}
By the assumption, given any $e\in\mathcal{E}$, the global encoding function $\phi_{ge}$ is Abelian group characterizable.
By Definition~\ref{def:groupFunc}, there exists a finite group $G$ with subgroups $\{G_f:f\in \mathcal{S}\cup e\}$ that characterize $\{X_f:f\in \mathcal{S}\cup e\}$.
By Definition~\ref{def:groupRV}, we define $g_r$ as a uniform random element in $G$, such that the coset $g_r G_\alpha$ represents $X_\alpha$ for any $\alpha\subseteq \mathcal{S}\cup e$.

In the context of group characterizable global encoding functions, let $(a_i G_i: i\in\mathcal{S})$ be any source message tuple.
By \cite{wei2017effect}, without loss of generality, we assume $|\cap_{i\in\mathcal{S}}G_i|=1$. 
By the assumption that source random variables $\{X_i:i\in\mathcal{S}\}$ are uniform and independent, we have $H(X_1,\dots,X_{|\mathcal{S}|})=\sum_{i\in\mathcal{S}} H(X_i)$ and $H(X_\alpha)=\log\frac{|G|}{|G_\alpha|}$ for any $\alpha\subseteq \mathcal{S}$.
Thus,
\begin{equation} \label{equa:acwl_0}
    \frac{|G|}{|G_\mathcal{S}|} = \prod_{i\in\mathcal{S}} \frac{|G|}{|G_i|}.
\end{equation}

We first show that for every source message tuple $(a_i G_i:i\in\mathcal{S})$ it holds that $\cap a_i G_i \neq \phi$. 
Each subgroup $G_i$ has $\frac{|G|}{|G_i|}$ cosets, such that the total number of tuples $|\{(a_i G_i:i\in\mathcal{S}): a_i\in G\}|$ equals $\prod_{i\in\mathcal{S}} \frac{|G|}{|G_i|}$.
Since by the definition each coset of $G_\mathcal{S}$ (there has $\frac{|G|}{|G_\mathcal{S}|}$ cosets in total) is a non-empty intersection corresponding to a message tuple, there are $\frac{|G|}{|G_\mathcal{S}|}$ non-empty intersections.
Thus, by (\ref{equa:acwl_0}), all intersections are non-empty, namely, $\cap_{i\in\mathcal{S}} a_i G_i\neq \phi$ for every $(a_i G_i: i\in\mathcal{S})$.

This implies that for every source message tuple $(a_i G_i: i\in\mathcal{S})$, it holds that $|\cap_{i\in\mathcal{S}} a_i G_i|=1$.
Therefore, there exists an element $a\in G$ such that $a=\cap_{i\in\mathcal{S}} a_i G_i$.
By Definition~\ref{def:groupFunc}, $\phi_{ge}$ outputs the coset of $G_e$ which includes $a$, such that
\begin{equation}\label{equa:acwl_1}
    \phi_{ge}(a_1 G_1,\dots,a_{|\mathcal{S}|} G_{|\mathcal{S}|}) = a G_e.
\end{equation}
Similarly, let $(b_i G_i)_{i\in\mathcal{S}}\neq (a_i G_i)_{i\in\mathcal{S}}$ be another source message tuple and let $b=\cap_{i\in\mathcal{S}} b_i G_i$; we have 
\begin{equation}\label{equa:acwl_2}
    \phi_{ge}(b_1 G_1,\dots, b_{|\mathcal{S}|} G_{|\mathcal{S}|}) = b G_e.
\end{equation}
For each $i\in\mathcal{S}$, as $a\in a_i G_i$, $b\in b_i G_i$, and $G$ is Abelian, $ab\in a_i G_i b_i G_i=a_i b_i G_i$ which implies $ab\in \cap_{i\in\mathcal{S}} (a_i b_i G_i)$.
Since $|G_\mathcal{S}|=1$ by assumption, we have $ab=\cap_{i\in\mathcal{S}} a_i b_i G_i$ which implies
\begin{equation}\label{equa:acwl_4}
    \phi_{ge}(a_1 b_1 G_1,\dots,a_{|\mathcal{S}|} b_{|\mathcal{S}|} G_{|\mathcal{S}|}) = a b G_e.
\end{equation}
By (\ref{equa:acwl_1}), (\ref{equa:acwl_2}), and (\ref{equa:acwl_4}), it holds that
\begin{equation*}
\begin{aligned}
    &\phi_{ge}(a_1 G_1,\dots,a_{|\mathcal{S}|} G_{|\mathcal{S}|})\cdot \phi_{ge}(b_1 G_1,\dots, b_{|\mathcal{S}|} G_{|\mathcal{S}|})\\ 
    =& \phi_{ge}(a_1 b_1 G_1,\dots,a_{|\mathcal{S}|} b_{|\mathcal{S}|} G_{|\mathcal{S}|}).    
\end{aligned}
\end{equation*}

By our assumption, $G$ is Abelian and all subgroups of Abelian group are normal.
By Proposition 7.11 in \cite{humphreys1996course}, $\{G/G_f:f\in\mathcal{S}\cup e\}$ are groups; thus $\phi_{ge}$ is a group homomorphism from the product of groups $G/G_1 \times \dots \times G/G_{|\mathcal{S}|}$ to the factor group $G/G_e$ (for formal definitions of concepts see \cite{gallian2016contemporary}).
Thus, $\phi_{ge}$ is Abelian CWL by Definition~\ref{def:cwlConsis}.

\end{proof}


The opposite direction representing CWL codes as group codes, was presented (under a slightly different set of definitions) in \cite{wei2019locala}.
The following theorem is given for completeness, and is proven in the Appendix.


\begin{theorem}\label{theorem:groupCWL5}
Let $\mathcal{I}$ be a network coding instance and $\{X_f: f\in\mathcal{S}\cup\mathcal{E}\}$ be a network code on $\mathcal{I}$. 
For any edge $e\in\mathcal{E}$, if the global encoding function $\phi_{ge}$ is CWL, then $\phi_{ge}$ is group characterizable.
\end{theorem}


One may attempt to extend Theorem~\ref{theorem:groupCWL4} to general group structures.
Our proof for Theorem~\ref{theorem:groupCWL4} will not extend directly, as subgroups of a non-Abelian group are not necessarily normal. 
That is, the left cosets corresponding to a given subgroup do not necessarily form a group, and thus cannot be used in the definition of CWL functions. 
Nevertheless, this does not imply that rates achievable using group codes cannot be obtained (or approached) by potentially different codes with global CWL functions. 
This latter problem is left for future study.


\section{CWL Network Codes}\label{sec:cwl}

We now study the connection between local and global CWL functions. 
The proofs for Theorem~\ref{theorem:CWL6} and Lemma~\ref{lemma:CWL2} below appear in the Appendix.


\begin{theorem}\label{theorem:CWL6}
Let $\mathcal{I}$ be a network coding instance and $\{X_f: f\in\mathcal{S}\cup\mathcal{E}\}$ be a network code on $\mathcal{I}$.
Assuming $\{X_f: f\in\mathcal{S}\cup\mathcal{E}\}$ can be obtained by local CWL encoding functions which are defined in a consistent manner, then $\{X_f: f\in\mathcal{S}\cup\mathcal{E}\}$ can be obtained through global CWL encoding functions which are defined in a consistent manner.
\end{theorem}


\begin{lemma}\label{lemma:CWL2}
Let $(\mathcal{X},\cdot)$, $(\mathcal{Y},\circ)$ be finite groups and $\bar{\mathcal{X}}$ a subgroup of $\mathcal{X}$, where $\mathcal{X}$ is Abelian.
If there exists a group homomorphism $\bar{\phi}: \bar{\mathcal{X}}\mapsto \mathcal{Y}$, then there exists a group homomorphism $\phi: \mathcal{X}\mapsto \mathcal{Y}$ such that $\phi(x)=\bar{\phi}(x)$ for $x\in  \bar{\mathcal{X}}$.
\end{lemma}


\begin{theorem}\label{theorem:CWL7}
Let $\mathcal{I}$ be a network coding instance and $\{X_f: f\in\mathcal{S}\cup\mathcal{E}\}$ be a network code on $\mathcal{I}$.
Assuming $\{X_f: f\in\mathcal{S}\cup\mathcal{E}\}$ can be obtained through global encoding functions which are Abelian CWL and defined in a consistent manner, then $\{X_f: f\in\mathcal{S}\cup\mathcal{E}\}$ can be obtained by local encoding functions which are Abelian CWL and defined in a consistent manner.
\end{theorem}

\begin{proof}
For any edge $e^*=(u,v)\in\mathcal{E}$, by the assumption in the theorem, the global encoding functions $\{\phi_{ge}:e\in\In(u)\cup e^*\}$ are CWL and $\{\mathcal{X}_f:f\in \mathcal{S}\cup \In(u)\cup e^*\}$ are finite groups.
Let $\phi_{le^*}: \mathcal{X}_{\In(u)}\mapsto \mathcal{X}_{e^*}$ be the local encoding function on $e^*$.

We define the function $\phi_{\In(u)}$ which maps from $\mathcal{X}_\mathcal{S}$ to $\mathcal{X}_{\In(u)}$ by $\phi_{\In(u)}(X_\mathcal{S})=(\phi_{ge}(X_\mathcal{S})|e\in\In(u))$.
For any $x_\mathcal{S}, x_\mathcal{S}'\in \mathcal{X}_\mathcal{S}$, by our definitions we have $\phi_{\In(u)}(x_\mathcal{S}\cdot x_\mathcal{S}') = \phi_{\In(u)}(x_\mathcal{S})\cdot \phi_{\In(u)}(x_\mathcal{S}')$ such that $\phi_{\In(u)}$ is a group homomorphism.
Here, the product is done component-wise accordingly to the operation on groups $\{\mathcal{X}_f: f\in \In(u)\}$.
By the Properties of Subgroups Under Homomorphisms \cite{gallian2016contemporary}, $\phi_{\In(u)}(\mathcal{X}_\mathcal{S})$ is a subgroup of $\mathcal{X}_{\In(u)}$.
We define $\bar{\mathcal{X}}_{\In(u)} = \phi_{\In(u)}(\mathcal{X}_\mathcal{S})$.
Namely, $\bar{\mathcal{X}}_{\In(u)}$ consists of all edge message tuples that appear in the communication (of some source information).

We define the function $\bar{\phi}_{le^*}:\bar{\mathcal{X}}_{\In(u)}\mapsto \mathcal{X}_{e^*}$ such that
\begin{equation*}
\begin{aligned}
    X_{e^*} 
    =& \bar{\phi}_{le^*}(X_{e_1},\dots, X_{e_{|\In(u)|}})\\
    =& \bar{\phi}_{le^*}(\phi_{ge_1}(X_{\mathcal{S}}), \dots, \phi_{ge_{|\In(u)|}}(X_{\mathcal{S}}))\\
    =& \phi_{ge^*}(X_{\mathcal{S}}).
\end{aligned}
\end{equation*}

Given any $x_{\In(u)}, x_{\In(u)}'\in \bar{\mathcal{X}}_{\In(u)}$, there exist source message tuples  $x_{\mathcal{S}}, x_{\mathcal{S}}'\in\mathcal{X}_\mathcal{S}$ such that $x_{\In(u)} = \phi_{\In(u)}(x_\mathcal{S})$ and $x_{\In(u)}' = \phi_{\In(u)}(x_\mathcal{S}')$.

On one hand,
\begin{equation*}
\begin{aligned}
    & \bar{\phi}_{le^*}(x_{e_1}x_{e_1}', \dots, x_{e_{|\In(u)|}} x_{e_{|\In(u)|}}' )\\
    =& \bar{\phi}_{le^*}(\phi_{ge_1}(x_{\mathcal{S}}) \cdot \phi_{ge_1}(x_{\mathcal{S}}'), \dots,\\ &\phi_{ge_{|\In(u)|}}(x_{\mathcal{S}}) \cdot \phi_{ge_{|\In(u)|}}(x_{\mathcal{S}}'))\\
    =& \bar{\phi}_{le^*}(\phi_{ge_1}(x_{\mathcal{S}}\cdot x_{\mathcal{S}}'), \dots, \phi_{ge_{|In(t)|}}(x_{\mathcal{S}}\cdot x_{\mathcal{S}}'))\\
    =& \phi_{ge^*}(x_{\mathcal{S}}\cdot x_{\mathcal{S}}').
\end{aligned}
\end{equation*}
On the other hand,
\begin{equation*}
\begin{aligned}
    & \bar{\phi}_{le^*}(x_{e_1}, \dots, x_{e_{|\In(u)|}})\cdot \bar{\phi}_{le^*}(x_{e_1}', \dots, x_{e_{|\In(u)|}}')\\
    =& \phi_{ge^*}(x_{\mathcal{S}})\cdot \phi_{ge^*}(x_{\mathcal{S}}').
\end{aligned}
\end{equation*}

Since $\phi_{ge^*}$ is CWL, $ \phi_{ge^*}(x_{\mathcal{S}}\cdot x_{\mathcal{S}}') = \phi_{ge^*}(x_{\mathcal{S}})\circ \phi_{ge^*}(x_{\mathcal{S}}')$, such that for any $x_{\In(u)}, x_{\In(u)}'\in \bar{\mathcal{X}}_{\In(u)}\subseteq  \mathcal{X}_{\In(u)}$, we have
\begin{equation*}
    \bar{\phi}_{le^*}\left( x_{\In(u)}\right) \circ \bar{\phi}_{le^*}\left(x'_{\In(u)}\right) = \bar{\phi}_{le^*}\left(x_{\In(u)}\cdot x'_{\In(u)}\right)
\end{equation*}
implying $\bar{\phi}_{le^*}$ is a group homomorphism from $\bar{\mathcal{X}}_{\In(u)}$ to $\mathcal{X}_{e^*}$.

Combining the facts that $\mathcal{X}_{\In(u)}$, $\mathcal{X}_{e^*}$ are Abelian groups and $\bar{\mathcal{X}}_{\In(u)}$ is a subgroup of $\mathcal{X}_{\In(u)}$, by Lemma~\ref{lemma:CWL2}, there exists a CWL function $\phi_{le^*}:\mathcal{X}_{\In(u)}\mapsto \mathcal{X}_{e^*}$ such that $\phi_{le^*}(x_{\In(u)})=\bar{\phi}_{le^*}(x_{\In(u)})$ for $x_{\In(u)}\in \bar{\mathcal{X}}_{\In(u)}$.
\end{proof}


We note that Theorem~\ref{theorem:CWL6} holds with respect to general CWL functions however Theorem~\ref{theorem:CWL7} only holds with respect to Abelian CWL functions. 
The challenge in proving Theorem~\ref{theorem:CWL7} for general CWL functions lies in extending Lemma~\ref{lemma:CWL2} to the case in which $\mathcal{X}$ is a non-Abelian group. 
In other words, given a partial function from $\mathcal{A}$ to $\mathcal{B}$ which is a group homomorphism $\phi': \mathcal{A}'\mapsto \mathcal{B}$, where $\mathcal{A}'\subset \mathcal{A}$, there may not exist a total function $\phi: \mathcal{A}\mapsto \mathcal{B}$ where $\phi$ is a group homomorphism. 

Combining Theorems~\ref{theorem:group2}, \ref{theorem:groupCWL4}, \ref{theorem:groupCWL5}, \ref{theorem:CWL6}, and \ref{theorem:CWL7} we conclude the following corollary. 
We remark that the reductions in Theorems~\ref{theorem:group2}, \ref{theorem:groupCWL5}, and \ref{theorem:CWL6} that were proven for general groups, preserve the Abelian group structures when used with Abelian groups.


\begin{corollary} \label{col:only}
Let $\mathcal{I}$ be a network coding instance and $\{X_f :f \in \mathcal{S} \cup \mathcal{E}\}$ be a network code on $\mathcal{I}$. 
$\{X_f :f \in \mathcal{S} \cup \mathcal{E}\}$ can be obtained by an Abelian group network code if and only if $\{X_f :f \in \mathcal{S} \cup \mathcal{E}\}$ can be obtained through local encoding functions which are Abelian CWL and defined in a consistent manner.
\end{corollary}


\section{Conclusion}\label{sec:conclusion}

In this work, we pursue an operational definition of group network codes. 
Through the study of CWL functions we show an equivalence between Abelian group network codes and codes obtained through local encoding functions which are Abelian CWL.
A number of questions are left open in this work. 
Primarily, the potential characterization of group network codes (in the non-Abelian case) through local CWL functions is left open. 
Leaving open the question whether CWL functions suffice to achieve the network coding capacity.


\section*{Acknowledgment}

Work supported in part by NSF grants CCF-1817241, CCF-1526771 and CCF-1909451.


\newpage
\appendix

\subsection{Proof of Theorem~\ref{theorem:linear1}}

\begin{proof}

We first assume a family of global functions for $\{X_f: f\in\mathcal{S}\cup\mathcal{E}\}$.
Consider an edge $e\in\mathcal{E}$, for any $e'\in \In(u)\cup e$ the edge message $x_{e'}$ is a linear function of the source message tuple $x_\mathcal{S}$, such that $x_{e'}=x_\mathcal{S} N_{e'}$.
Here $N_{e'}$ is the global encoding matrix on $e'$.
The local encoding function $\phi_{le'}$ on $e'$ is linear, if and only if there exist local encoding matrices $\{M_{e'}:e'\in \In(u)\}$ such that
\begin{equation}\label{equa:linear1}
    x_e= \sum_{e'\in \In(u)} x_{e'}M_{e'}.
\end{equation}
That is, if
\begin{equation}\label{equa:linear2}
    x_\mathcal{S}N_{e}= \sum_{e'\in \In(u)} x_\mathcal{S}N_{e'}M_{e'}
\end{equation}
for any $x_\mathcal{S}\in \mathbb{F}_q^{1\times n}$.

Equation (\ref{equa:linear2}) holds if there exist local encoding matrices $\{M_{e'}:e'\in \In(u)\}$ such that
\begin{equation} \label{equa:linear3}
\begin{aligned}
    N_{e} &= \sum_{e'\in \In(u)} N_{e'}M_{e'}\\ &= 
    \begin{bmatrix}
        N_{e_1}\dots N_{e_{|\In(u)|}}
    \end{bmatrix}
    \begin{bmatrix}
        M_{e_1}\\ \vdots\\ M_{e_{|\In(u)|}}
    \end{bmatrix}.    
\end{aligned}
\end{equation}
We denote $ \begin{bmatrix} N_{e_1}\dots N_{e_{|\In(u)|}} \end{bmatrix}$ as $ N_{\In(u)}$ and $ \begin{bmatrix} N_{e_1}\dots N_{e_{|\In(u)|}}, N_e \end{bmatrix}$ as $ N_{\In(u), e}$.
If we consider the entries of $\{M_{e'}:e'\in \In(u)\}$ as unknowns, then (\ref{equa:linear3}) is a system of linear equations,
which has solution if and only if 
\begin{equation}\label{equa:linear4}
    rank(N_{\In(u), e})=rank(N_{\In(u)}).
\end{equation}

For $\alpha\subseteq \In(u)\cup e$, we denote $\begin{bmatrix} N_{e_1}\dots N_{e_{|\alpha|}} \end{bmatrix}$ as $N_\alpha$. 
Since $X_\mathcal{S}$ is uniformly distributed over $\mathcal{X}_\mathcal{S}=\mathbb{F}_2^{1\times n R_\mathcal{S}}$, where $R_\mathcal{S}=\sum_{i\in\mathcal{S}} R_i$, 
\begin{equation*}
\begin{aligned}
     \Pr(X_\alpha=x_\alpha) =& \sum_{x_\mathcal{S}:x_\alpha=x_\mathcal{S}N_{\alpha}} \Pr(X_\mathcal{S}=x_\mathcal{S})\\
     =& \frac{|\{x_\mathcal{S}\in \mathcal{X}_\mathcal{S}: x_{\alpha}=x_\mathcal{S}N_{\alpha}\}|}{|\mathcal{X}_{\mathcal{S}}|}\\
     =& \frac{2^{nR_\mathcal{S}-rank(N_{\alpha})}}{2^{n R_\mathcal{S}}}\\
     =& 2^{-rank(N_{\alpha})},
\end{aligned}
\end{equation*}
such that $X_{\alpha}$ is uniform and 
\begin{equation}\label{equa:linear5}
    H(X_{\alpha})=rank(N_{\alpha}).
\end{equation}

Because $X_e$ is a function of $X_{\In(u)}$, we have $H(X_e| X_{\In(u)}) = 0$ which implies $H(X_e, X_{\In(u)}) = H(X_{\In(u)})$.
By (\ref{equa:linear5}), we have $rank(N_{\In(u), e})=rank(N_{\In(u)})$.
Thus (\ref{equa:linear4}) is true which concludes the proof of this direction. 
The other direction (from local functions to global ones) is proven by induction. See, e.g. \cite{yeung2008information}.

\end{proof}


\subsection{Proof of Theorem~\ref{theorem:group1}}

\begin{proof}
We show the if-part first.
By the assumption in the theorem, for every $e=(u,v)\in\mathcal{E}$, the set of random variables $\{X_e,X_{\In(u)}\}$ are group characterizable in a consistent manner.
As $\mathcal{S}\cup\mathcal{E}=\cup_{e\in\mathcal{E}} \{\In(u)\cup e\}$, by Lemma~\ref{lemma:groupunionofRV} the set of random variables $\{X_f: f\in\mathcal{S}\cup \mathcal{E}\}$ is group characterizable.
Thus, by Definition~\ref{def:groupNC} the network code is a group network code.

For the only-if part, by Definition~\ref{def:groupNC} all local encoding functions are group characterizable, which concludes the proof.
\end{proof}


\subsection{Proof of Theorem~\ref{theorem:group2}}

\begin{proof}
By the assumption in the theorem, for every $e=(u,v)\in\mathcal{E}$, the set of random variables $\{X_f:f\in \mathcal{S}\cup e\}$ are group characterizable in a consistent manner. 
As $\mathcal{S}\cup\mathcal{E}=\cup_{e\in\mathcal{E}} \{\mathcal{S}\cup e\}$,
the theorem can be proven following the analysis of Theorem~\ref{theorem:group1}.
\end{proof}


\subsection{Proof of Theorem~\ref{theorem:groupCWL5}}

\begin{proof}

By our definitions, to prove that $\phi_{ge}$ is group characterizable, it suffices to show that there exists a finite group $G$ with subgroups $\{G_f: f\in\mathcal{S}\cup e\}$ which forms a group characterization of random variables $\{X_f: f\in \mathcal{S}\cup e\}$.

By the assumption in the theorem, the function $\phi_{ge}$ is CWL. 
Thus there exists finite groups $\{H_f: f\in \mathcal{S}\cup e\}$, where $\mathcal{X}_f= H_f$ for $f\in\mathcal{S}\cup e$ and $\phi_{ge}$ is a group homomorphism from $H_1\times \dots \times H_{|\mathcal{S}|}$ to $H_e$.

We define $G= H_1\times \dots \times H_{|\mathcal{S}|}$, $G_e=\ker(\phi_{ge})$ and $G_i= H_1\times\dots \times \{\mathbf{i}_i\}\times \dots\times H_{|\mathcal{S}|}$ for each $i\in\mathcal{S}$ (replacing $H_i$ by $\mathbf{i}_i$ on the $i$-th coordinate of $G_i$).
By our definitions, for any $g=(h_1,\dots,h_{|\cS|})\in G$ and $i\in\mathcal{S}$, it holds that $g G_i = H_1\times \dots \times \{h_i\} \times \dots \times H_{\cS} = G_i g$ such that $\{G_i: i\in\mathcal{S}\}$ are normal subgroups of $G$.
By the First Isomorphism Theorem \cite{gallian2016contemporary}, as $\phi_{ge}$ is a group homomorphism from $G$ to $H_e$, $\ker(\phi_{ge})$ is a normal subgroup of $G$.
Thus, $G$ is a finite group and $\{G_f: f\in\mathcal{S}\cup e\}$ are normal subgroups of $G$. 
Additionally, as $\{G_f: f\in\mathcal{S}\cup e\}$ are normal, by the Factor Groups Theorem in \cite{gallian2016contemporary}, the set of left cosets $G/G_f=\{g G_f: g\in G\}$ is a group under operation $(a G_f)(b G_f)=ab G_f$ for each $f\in\mathcal{S}\cup e$.

Now, we show that $G/G_f$ is isomorphic to $H_f$ for any $f\in\mathcal{S}\cup e$.
Given any $i\in\mathcal{S}$, one can construct a function $\psi_i: H_i\mapsto G/G_i$, where $\psi_i(h_i)= (\mathbf{i}_1,\dots,h_i,\dots,\mathbf{i}_{|\mathcal{S}|}) G_i$.
As for any $h_i,h_i'\in H_i$, we have
\begin{equation*}
    \psi_i(h_i\cdot h_i')=\psi_i(h_i)\cdot \psi_i(h_i'),
\end{equation*}
thus $\psi_i$ is a group homomorphism.
According to our definition, we have $\ker(\psi_i)=\{h_i\in H_i: \psi_i(h_i)= G_i\}=\mathbf{i}_i$, thus by the Properties of Subgroups Under Homomorphisms \cite{gallian2016contemporary}, $\psi_i$ is an isomorphism from $H_i$ to $G/G_i$ which implies that $H_i$ and $G/G_i$ are isomorphic.
We now show that $G/G_e$ is isomorphic to $H_e$.
One can construct a function $\psi_e: G/G_e\mapsto H_e$, where $\psi_e(g G_e)=\phi_{ge}(g)=h_e$ for any $g\in G$.
As for any $g G_e, g'G_e\in G/G_e$, we have
\begin{equation*}
    \begin{aligned}
        \psi_e(g G_e\cdot g'G_e) =& \phi_{ge}(g\cdot g')=\phi_{ge}(g)\cdot \phi_{ge}(g')\\ =&\psi_e(g G_e)\cdot \psi_e(g'G_e),
    \end{aligned}
\end{equation*} 
thus $\psi_e$ is a group homomorphism, and $\ker(\psi_e)=\ker(\phi_{ge})=G_e$ which is the identity element in $G/G_e$.
Similar to the argument above, $\psi_e$ is an isomorphism from $G/G_e$ to $H_e$.

Next, we show that the group $G$ with subgroups $\{G_f: f\in\mathcal{S}\cup e\}$ forms a group characterization of $\{X_f: f\in\mathcal{S}\cup e\}$.
More specifically, we present a collection of random variables $\{Y_f: f\in\mathcal{S}\cup e\}$ that are identically distributed to $\{X_f: f\in\mathcal{S}\cup e\}$ which are group characterized by $G$ with subgroups $\{G_f: f\in\mathcal{S}\cup e\}$.
The random variables $\{Y_f: f\in\mathcal{S}\cup e\}$ are defined by $(X_1, \dots, X_{|\mathcal{S}|})$ which are distributed uniformly over $H_1 \times \dots \times H_{|\mathcal{S}|} = G$. 
Namely, for $f\in\mathcal{S}\cup e$, let $Y_f=(X_1,\dots,X_{|\mathcal{S}|})G_f$.
Using the definition above, $(X_1,\dots,X_{|\mathcal{S}|}, X_e) = (h_1,\dots,h_{|\mathcal{S}|},h_e)$ if and only if for all $f \in \mathcal{S}\cup e$, $Y_f = (h_1,\dots,h_{|\mathcal{S}|}) G_f$. 
More specifically, it holds for all $i \in \mathcal{S}$ that $Y_i=\psi_f(X_i)$ and $X_e=\psi_e(Y_e)$, where $\psi_f$ ($f\in\mathcal{S}\cup e$) is the isomorphism discussed above.
Thus, $\{Y_f: f\in\mathcal{S}\cup e\}$ are identically distributed to $\{X_f: f\in\mathcal{S}\cup e\}$ and in addition, for the uniform element $g_r=(X_1,\dots,X_{|\mathcal{S}|})$ in $G$, for all $f \in \mathcal{S}\cup e$, $Y_f = g_r G_f$. 
This concludes the proof.

\end{proof}


\subsection{Proof of Theorem~\ref{theorem:CWL6}}

\begin{proof}

We start by noticing that if the encoding functions on all edges directly connected to sources are locally CWL, then they are also globally CWL. 

For any edge $e^*=(u,v)\in\mathcal{E}$, assuming by induction that global functions $\{\phi_{ge}: e\in In(u)\}$ and local function $\phi_{le^*}$ are CWL with a set of groups $\{G_f: f\in\mathcal{S}\cup \In(u)\cup e^*\}$.
Let $\phi_{ge^*} (x_{\mathcal{S}}) = \phi_{le^*}(\phi_{e_1}(x_{\mathcal{S}}), \dots, \phi_{e_{|In(e^*)|}}(x_{\mathcal{S}}))$.
Let $x_{\mathcal{S}}\neq x_{\mathcal{S}}'$ be source message tuples.
Define $x_{\mathcal{S}}\cdot x_{\mathcal{S}}'=(x_1\cdot x_1',\dots,x_{|\mathcal{S}|}\cdot x_{|\mathcal{S}|}')$.
By our definitions, 
\begin{equation*}
\begin{aligned}
    &\phi_{ge^*}(x_{\mathcal{S}}\cdot x_{\mathcal{S}}')\\ 
    =& \phi_{le^*}\left(\phi_{e_1}(x_{\mathcal{S}}\cdot x_{\mathcal{S}}'),\dots,\phi_{e_{|In(e^*)|}}(x_{\mathcal{S}}\cdot x_{\mathcal{S}}')\right)\\ 
    =& \phi_{le^*} (\phi_{e_1}(x_{\mathcal{S}})\cdot \phi_{e_1}(x_{\mathcal{S}}'), \dots,\\ &\phi_{e_{|In(e^*)|}}(x_{\mathcal{S}}) \cdot \phi_{e_{|In(e^*)|}} (x_{\mathcal{S}}') )\\ 
    =& \phi_{le^*}(\phi_{e_1}(x_{\mathcal{S}}),\dots,\phi_{e_{|In(e^*)|}}(x_{\mathcal{S}}))\\ &\cdot \phi_{le^*}(\phi_{e_1}(x_{\mathcal{S}}'),\dots,\phi_{e_{|In(e^*)|}}(x_{\mathcal{S}}'))\\ 
    =& \phi_{ge^*}(x_{\mathcal{S}})\cdot \phi_{ge^*}(x_{\mathcal{S}}').
\end{aligned}
\end{equation*}
Thus, $\phi_{ge^*}$ is CWL.
Continuing inductively, since $\mathcal{I}$ is a directed acyclic network instance, the global encoding function on any edge of the network is CWL.

\end{proof}


\subsection{Proof of Lemma~\ref{lemma:CWL2}}

\begin{proof}

Given a group homomorphism $\bar{\phi}:\bar{\mathcal{X}}\mapsto \mathcal{Y}$, the kernel $\ker(\bar{\phi})$ is a normal subgroup of $\bar{\mathcal{X}}$. 
By the Fundamental Theorem of Abelian groups, given a subgroup $\bar{\mathcal{X}}\subseteq \mathcal{X}$, there exists a subgroup $K\subseteq \mathcal{X}$ such that $\mathcal{X}$ equals to $K\cdot \bar{\mathcal{X}}=\{k\cdot \bar{x} : k\in K, \bar{x}\in \bar{\mathcal{X}}\}$.

Now, we define $\phi: \mathcal{X}\mapsto \mathcal{Y}$.
Since $\mathcal{X}=K\cdot \bar{\mathcal{X}}$, for any element $x\in\mathcal{X}$, there exists $k_x\in K$ and $\bar{x}\in \bar{\mathcal{X}}$ such that $x=k_x\cdot \bar{x}$.
We define $\phi(x)=\phi(k_x\cdot \bar{x})=\bar{\phi}(\bar{x})$.
Similarly, for any $x'\in \mathcal{X}$, let $x'=k_x'\cdot \bar{x}'$.
By our definitions, we have
\begin{equation*}
\begin{aligned}
    &\phi(x)\circ\phi(x')\\ =& \phi(k_x\cdot \bar{x})\circ \phi(k_x'\cdot \bar{x}')=\bar{\phi}(\bar{x})\circ \bar{\phi}(\bar{x}')\\=&\bar{\phi}(\bar{x}\cdot \bar{x}')=\phi(k_x\cdot \bar{x}\cdot k_x'\cdot \bar{x}')\\ =&\phi(x\cdot x')
\end{aligned}
\end{equation*}
which implies that $\phi$ is a group homomorphism from $\mathcal{X}$ to $\mathcal{Y}$.

\end{proof}
\bibliographystyle{IEEEtran}
\bibliography{bibliography}
\end{document}